\documentclass[a4paper]{article}

\usepackage{complexity} 
\usepackage{enumerate}
\usepackage{amsmath,amssymb,url}
\usepackage{amsthm}
\usepackage{euler}
\usepackage{authblk}
\usepackage{todonotes}
\usepackage{graphicx}
\usepackage{pifont}
\usepackage{geometry}
\usepackage{mathtools}
\usepackage{thmtools}
\usepackage{thm-restate}
\usepackage{hyperref}
\usepackage{cleveref}
\usepackage{authblk}
\usepackage{todonotes}
\usepackage{microtype}
\usepackage{multirow}
\usepackage{amsmath,amssymb}
\usepackage{comment}
\usepackage{enumerate}
\usepackage{xspace}
\usepackage{listings}
\usepackage{color}
\usepackage{xcolor}
\usepackage{cite}
\RequirePackage{fancyhdr}
 \usepackage{xcolor}
\usepackage{boxedminipage}
\usepackage{algorithm2e}

\newcommand{\mypara}[1]{\smallskip\noindent{\textbf{\sffamily #1} \ }}

\newcommand{\nats}{{\mathbb{N}}}

\newcommand{\cO}{{\cal O}}

\newcommand{\Yes}{\textsf{yes}}

\newcommand{\cF}{{\mathcal{F}}}

\newcommand{\cR}{\mathcal{R}}

\newcommand{\cA}{\mathcal{A}}

\newcommand{\cH}{\mathcal{H}}

\newcommand{\cT}{\mathcal{T}}

\newcommand{\cM}{\mathcal{M}}
\newcommand{\cN}{\mathcal{N}}

\newcommand{\OPT}{\mbox{\rm OPT}}

\newcommand{\eps}{\varepsilon}

\newcommand{\cJ}{{\mathcal J}}

\newcommand{\shortversion}[1]{}

\newcommand{\CBT}{{{\textsf{CON-}}{\eta}{\textsf{-TDS}}}}

\newcommand{\CVC}{{\sc Connected Vertex Cover}}

\newcommand{\OO}{{\mathcal O}}
\newcommand{\bigoh}{{\mathcal O}}

\newcommand{\nka}{${\sf NP \subseteq coNP/poly}$}
\newcommand{\nn}{{\mathbb N}}
\newcommand{\rr}{{\mathbb R}}

\newenvironment{tightcenter}
 {\parskip=0pt\par\nopagebreak\centering}
 {\par\noindent\ignorespacesafterend}

\usepackage{tikz}
\usetikzlibrary{calc}
\usepackage{xargs}
\usepackage{xifthen}
\usepackage{framed}

\usepackage{ctable}

\newlength{\RoundedBoxWidth}
\newsavebox{\GrayRoundedBox}
   {\setlength{\RoundedBoxWidth}{\textwidth-4.5ex}
    \def\boxheading{#1}
    \begin{lrbox}{\GrayRoundedBox}
       \begin{minipage}{\RoundedBoxWidth}%
   }{%
       \end{minipage}
    \end{lrbox}%
    \begin{tightcenter}%
    \begin{tikzpicture}%
       \node(Text)[draw=black!60,fill=white,rounded corners,%
             inner sep=2ex,text width=\RoundedBoxWidth]%
             {\usebox{\GrayRoundedBox}};
        \coordinate(x) at (current bounding box.north west);
        \node [draw=white,rectangle,inner sep=3pt,anchor=north west,fill=white] 
        at ($(x)+(6pt,.75em)$) {\boxheading};
    \end{tikzpicture}
    \end{tightcenter}\vspace{0pt}%
    \ignorespacesafterend
}    

\newcommand{\defparproblem}[4]{
 \vspace{1mm}
\noindent\fbox{
  \begin{minipage}{.95\textwidth}
  \begin{tabular*}{\textwidth}{@{\extracolsep{\fill}}lr} \textsc{#1} \\ \end{tabular*}
  {\bf{Input:}} #2  \\
  {\bf{Parameter:}} #3\\
  {\bf{Question:}} #4
  \end{minipage}
  }
  \vspace{1mm}
}

\newcommand{\sv}[1]{}

\newcommand{\CBTDS}{{\sc Connected $\eta$-Treedepth Deletion}}
\newcommand{\CetaTran}{{\sc Con-$\eta$-Depth-Transversal}}
\newcommand{\cbtds}{{connected $\eta$-treedepth deletion set}}
\newcommand{\btds}{{$\eta$-treedepth deletion set}}

\newcommand{\td}{{\sf td}}

\newcommand{\uclos}{{\sf UClos}}

\newcommand{\ceil}[1]{\lceil #1 \rceil}
\newcommand{\floor}[1]{\lfloor #1 \rfloor}

\newcommand{\nice}{nice}
\newcommand*{\components}[1]{\ensuremath{\operatorname{Comp}({#1})}}
\newcommand*{\componentsTD}[2]{\ensuremath{\operatorname{Comp}({#1}, {#2})}}

\newenvironment{claimproof}[1][\proofname]{\proof[#1]}{\endproof}

\newtheorem{definition}{\bf Definition}
\newtheorem{proposition}{\bf Proposition}

\newtheorem{claim}{\bf Claim}

\newtheorem{observation}{\bf Observation}
\newtheorem{reduction rule}{\bf Reduction Rule}
\newtheorem{lemma}{\bf Lemma}

\title{An Improved Time-Efficient Approximate Kernelization for Connected Treedepth Deletion Set \footnote{A preliminary version \cite{EibenMR22} of this paper has appeared in proceedings of WG 2022.}}

\author[1]{Eduard Eiben}
\author[2]{Diptapriyo Majumdar}
\author[3]{M. S. Ramanujan}
\affil[1]{Royal Holloway, University of London, Egham, United Kingdom
    \texttt{eduard.eiben@rhul.ac.uk}}
\affil[2]{Indraprastha Institute of Information Technology Delhi, New Delhi, India
\texttt{diptapriyo@iiitd.ac.in}}
\affil[3]{University of Warwick, Coventry, United Kingdom
\texttt{R.Maadapuzhi-Sridharan@warwick.ac.uk}}

\date{}

\begin{document}

\maketitle


\begin{abstract}
We study the {\CBTDS} problem, where the input instance is an undirected graph $G$, and an integer $k$ and the objective is to decide whether there is a vertex set $S \subseteq V(G)$ such that $|S| \leq k$, every connected component of $G - S$ has treedepth at most $\eta$ and $G[S]$ is a connected graph.
As this problem naturally generalizes the well-studied {\CVC} problem, when parameterized by the solution size $k$, {\CBTDS} is known to not admit a polynomial kernel unless {\nka}. This motivates the question of designing {\em approximate polynomial kernels} for this problem. 

In this paper, we show that for every fixed $0 < \eps \leq 1$, {\CBTDS} admits a time-efficient $(1+\eps)$-approximate kernel of size $k^{2^{\OO(\eta + 1/\eps)}}$  
 (i.e., a Polynomial-size Approximate Kernelization Scheme).
\end{abstract}
\section{Introduction}

Parameterized complexity is a popular approach to cope with NP-Completeness and the related area of kernelization studies mathematical formulations of preprocessing algorithms for (typically) NP-complete decision problems.
Kernelization is an important step that preprocesses the input instance $(I, k)$ into a smaller,  equivalent instance $(I', k')$ in polynomial-time such that $|I|' + k'$ is bounded by $g(k)$.
It is desired that $g(k)$ is polynomial in $k$, in which case we have a {\em polynomial kernelization}.
Over the past few decades, the design of (polynomial) kernelization for numerous problems has been explored~\cite{CyganFKLMPPS15, Cygan12, CyganPPW13, MisraPRS14} and a rich variety of algorithm design techniques have been introduced.
There are, however, problems that provably do not admit polynomial kernels unless {\nka} \cite{BodlaenderJK14, DomLS14, HermelinKSWW15, HermelinW12}, in which case one requires an alternate rigorous notion of preprocessing.
Moreover, the notion of kernelization is defined with respect to decision problems, implying that when a suboptimal solution to the reduced instance is provided, one may not be able to get a feasible solution to the original input instance.
To address both of the aforementioned issues with kernelization, Lokshtanov et al.~\cite{LokshtanovPRS17} introduced the framework of Approximate Kernelization.
Roughly speaking, an $\alpha$-approximate kernelization is a polynomial-time preprocessing algorithm for a parameterized optimization problem with the promise that if a $c$-approximate solution to the reduced instance is given, then a $(c\cdot\alpha)$-approximate feasible solution to the original instance can be obtained in polynomial time. In this case, both $c, \alpha \geq 1$. When the reduced instance has size bounded by $g(k)$ for some polynomial function $g$, then we have a $\alpha$-approximate polynomial-size approximate kernel (see Section \ref{sec:prelims} for formal definitions).

In recent years, there has been a sustained search for polynomial-size approximate kernels for well-known problems in parameterized complexity that are known to exclude standard polynomial kernelizations. One such set of problems is the family of ``vertex deletion'' problems with a {\em connectivity constraint}.
A classic example here is {\sc Vertex Cover} that admits a $2k$ vertex kernel, but {\CVC} provably does not admit a polynomial kernel unless {\nka}.
Lokshtanov et al.~\cite{LokshtanovPRS17} proved that for every $\eps > 0$, {\CVC} admits a $(1+\eps)$-approximate kernel of size $\cO(k^{\ceil{1/\eps}})$. This is also called a {\em Polynomial-size Approximate Kernelization Scheme} (PSAKS). Subsequent efforts have mainly focused on studying the feasibility of approximate kernelization for problems that generalize {\CVC}. 
For instance, Eiben et al.~\cite{EibenHR19} obtained a PSAKS for the {\sc Connected $\cH$-Hitting Set} problem (where one wants to find a smallest connected vertex set that hits all occurrences of graphs from the finite set $\cH$ as induced subgraphs) and Ramanujan obtained a PSAKS for {\sc Connected Feedback Vertex Set}~\cite{Ramanujan19} and  a $(2+\eps)$-approximate polynomial {\em compression}~\cite{Ramanujan21} for the {\sc Planar $\cF$-deletion} problem~\cite{FominLMS12} with connectivity constraints on the solution. A compression is a weaker notion than kernelization, where the output is not required to be an instance of the original problem at hand. 

In this paper, our focus is on the connectivity constrained version of the {\sc $\eta$-Treedepth Deletion Set} problem. In the (unconnected version of the) problem, one is given a graph $G$ and an integer $k$ and the goal is to decide whether there is a vertex set of size at most $k$ whose deletion leaves a graph of treedepth at most $\eta$.  We refer the reader to Section~\ref{sec:prelims} for the formal definition of treedepth. Intuitively, it is a graph-width measure that expresses the least number of rounds required to obtain an edge-less graph, where, in each round we delete some vertex from each surviving connected component. Treedepth is a graph parameter that has attracted significant interest in the last decade. It allows improved algorithmic bounds over the better-known parameter of treewidth for many problems (see, for example, \cite{ReidlRVS14,HegerfeldK20}) and it plays a crucial role in the study of kernelization~\cite{GajarskyHOORRVS17}. In recent years, the optimal solution to the {\sc $\eta$-Treedepth Deletion Set} problem itself has been identified as a useful parameter in the kernelization of generic vertex-deletion problems~\cite{JansenP20}.  The many insightful advances made by focusing on graphs of bounded treedepth motivates us to consider the {\CBTDS} problem as an ideal conduit between the well-understood {\CVC} problem and the connected versions of more general problems such as the {\sc $\eta$-Treewidth Deletion Set} problem, which is still largely unexplored from the point of view of approximate kernelization. 
We formally state our problem as follows.

\defparproblem{{\CBTDS} (\CetaTran)}
	{An undirected graph $G$, and an integer $k$.}{$k$}{Does $G$ have a set $S$ of at most $k$ vertices such that $G[S]$ is connected and $(G - S)$ has treedepth at most $\eta$?}

 A set $S \subseteq V(G)$ is called a {\em {\cbtds}} if $G[S]$ is connected and every connected component of $G - S$ has treedepth at most $\eta$.
As edgeless graphs have treedepth 1, it follows that {\CBTDS} generalizes {\CVC} and does not have a polynomial kernelization under standard hypotheses even for constant values of $\eta$, thus motivating its study through the lens of approximate kernelization. 
Here, two results in the literature are of particular consequence to us and form the starting point of our work:

\begin{itemize}
	\item  Graphs of treedepth at most $\eta$ can be characterized by a finite set of forbidden induced subgraphs, where each obstruction has size at most $2^{2^{\eta - 1}}$~\cite{DvorakGT12} and hence, an invocation of the result of Eiben et al.~\cite{EibenHR19} gives a $(1+\eps)$-approximate kernelization of size $\cO(k^{2^{2^{\eta -1}}\cdot2^{\frac{1}{\eps}}+1})$ for {\CetaTran}.
	\item On the other hand, using the fact that graphs of treedepth at most $\eta$ also exclude a finite set of graphs as forbidden minors including at least one planar graph, we infer that the $(2+\eps)$-approximate polynomial compression for {\sc Connected Planar $\cF$-Deletion} of Ramanujan~\cite{Ramanujan21} implies a $(2+\eps)$-approximate compression for {\CBTDS} of size $k^{f(\eta)\cdot 2^{\OO(1/\eps)}}$ for some function $f$ that is at least exponential.
\end{itemize}

Naturally, these two ``meta-results'' provide useful proofs of concept using which we can conclude the existence of an approximate kernel (or compression) for {\CBTDS}. However, the kernel-size bounds that one could hope for by taking this approach are far from optimal and in fact, the second result mentioned above only guarantees the weaker notion of compression.  Thus, these two results  raise the following natural question: ``Could one exploit structure inherent to the bounded treedepth graphs and improve upon both results, by obtaining a $(1+\eps)$-approximate polynomial kernelization for {\CBTDS} with improved size bounds?"
Our main result is a positive answer to this question. 

\begin{restatable}{theorem}{thmpsaksCBTDSmain}\label{thm:psaks-CBTDS-main}
For every fixed $0<\eps \leq 1$, {\CBTDS} has a time-efficient $(1+\eps)$-approximate kernelization of size $k^{\OO (\ceil{{2^{\eta + \ceil{10/\eps}}\eta}/{\eps}})}$.
\end{restatable}


\section{Preliminaries}
\label{sec:prelims}

\mypara{Sets and Graphs:}
We use $[r]$ to denote the set $\{1,\ldots,r\}$ and $A \uplus B$ to denote the disjoint union of two sets.
We use standard graph theoretic terminologies from Diestel's book~\cite{DiestelBook}.
Throughout the paper, we consider undirected graphs.
We use $P_{\ell}$ to denote a path with $\ell$ vertices.
A graph is said to be {\em connected} if there is a path between every pair of vertices.
Let $G = (V, E)$ be a graph and a pair of vertices $u, v \in V(G)$.
We call a set $A \subseteq V(G)$ an {\em $(u, v)$-vertex cut} if there is no path from $u$ to $v$ in $G - A$.
For $u, v \in G$, we use $dist(u, v)$ to denote the length of a `shortest path' from $u$ to $v$.
We use $diam(G) = \max_{u, v \in V(G)} dist(u, v)$ to denote the {\em diameter} of $G$.
Let $R \subseteq V(G)$ be a vertex set the elements of which are called {\em terminals} and a weight function $w: E(G) \rightarrow \nats$. 
A {\em Steiner tree} with terminal set $R$ is a subgraph $T$ of $G$ such that $T$ is a tree and $R \subseteq V(T)$. The \emph{weight} of a Steiner tree $T$ is $w(T) = \sum\limits_{e \in E(T)} w(e)$.
A {\em $t$-component} for $R$ is a tree with at most $t$ leaves and all these leaves coincide with a subset of $R$.
A {\em $t$-restricted Steiner tree} for $R$ is a collection $\cT$ of $t$-components for $R$ such that the union of the $t$-components in $\cT$ induces a Steiner tree for $R$. We refer to Byrka et al. \cite{ByrkaGRS13} for more detailed introduction on these terminologies.

\begin{proposition}[\cite{BorchersD97}]
	\label{proposition:approx-k-steiner-tree}
	For every $t \geq 1$, given a graph $G$, a terminal set $R$, a cost function $w: E(G) \rightarrow \nats$, and a Steiner tree $T$ for $R$, there exists a $t$-restricted Steiner tree $\cT$ for $R$ of cost at most $(1+ \frac{1}{\floor{\log_2 t}})\cdot w(T)$.
\end{proposition}

\begin{proposition}[\cite{DreyfusW71}]
	\label{proposition:steiner-tree-terminals}
	Let $G$ be a graph, $R$ be a set of terminals, and a $w: E(G) \rightarrow \nats$ be a cost function.
	Then, a minimum weight Steiner tree for $R$ can be computed in $\cO(3^{|R|}|V(G)||E(G)|)$-time.
\end{proposition}
Note that if \(|R|\) is constant then the above algorithm runs in polynomial time.

\mypara{Treedepth:} Given a graph $G$, we define ${\td}(G)$, the {\em treedepth} of $G$ as follows.

\begin{equation}
\label{eq:treedepth-definition}
	{\td} (G) =
	\begin{cases}
		1 & \mbox{ if }|V(G)| = 1\\
		1 + \min\limits_{v \in V(G)} {\td}(G - v) & \mbox{ if } G \mbox{ is connected and }|V(G)| > 1 \\
		\max\limits_{i = 1}^p {\td}(G_i) & \mbox{ if }G_1,\ldots,G_p \mbox{ are connected components of }G
	\end{cases}
\end{equation}

A \emph{treedepth decomposition} of graph $G = (V, E)$ is a rooted forest $Y$ with vertex set
$V$, such that for each edge $uv \in E(G)$, we have either that $u$ is an ancestor of $v$ or $v$ is an
ancestor of $u$ in $Y$.
Note that a treedepth decomposition of a connected graph $G$ is equivalent to some depth-first search tree of $G$ and in the context of treedepth is also sometimes referred to as {\em elimination tree} of $G$.
It is clear from the definition that the treedepth of a graph $G$ is equivalent to the minimum depth of a treedepth decomposition of $G$, where depth is defined as the maximum number of vertices along a path from the root of the tree to a leaf \cite{NesetrilM06}.
Let $T$ be a tree rooted at a node $r$.
The {\em upward closure} for a set of nodes $S \subseteq V(T)$ is denoted by $\uclos_T (S) = \{v \in V(T)\mid v$ is an ancestor of $u\in S$ in $T\}$. This notion has proved useful in the kernelization algorithm of Giannopoulou et al.~\cite{GiannopoulouJLS17} for \(\eta\)-\textsc{Treedepth Deletion}.  


The following facts about the treedepth of a graph will be useful throughout the paper. 

\begin{proposition}[\cite{NesetrilM06}]
	\label{lemma:bounded-treedepth-diameter}
	Let $G$ be a graph such that $\td (G) \leq \eta$.
	Then, the diameter of $G$ is at most $2^{\eta}$.
\end{proposition}

\begin{proposition}[\cite{FominLMS12}]
\label{prop:approximation}
	For every constant \(\eta\in \mathbb{N}\), there exists a polynomial-time \(\cO(1)\)-approximation for \(\eta\)-\textsc{Treedepth Deletion}.
\end{proposition}

\begin{proposition}[\cite{ReidlRVS14}]\label{proposition:computing_elimination_tree}
	Let $G$ be a connected graph and \(\eta\in \mathbb{N}\). There exists an algorithm running in \(\bigoh(f(\eta)|V(G)|)\)-time, for some computable function \(f\), that either correctly concludes that \(\td(G) > \eta\) or computes a treedepth decomposition for \(G\) of depth at most \(\eta\). 
\end{proposition}

\mypara{Parameterized algorithms and kernels:} A parameterized problem $\Pi$ is a subset of $\Sigma^* \times \nn$ for a finite alphabet $\Sigma$.
An instance of a parameterized problem is a pair $(x, k)$ where $x \in \Sigma^*$ is the input and $k \in \nn$ is the parameter.
We assume without loss of generality that $k$ is given in unary.
We say that $\Pi$ admits a {\em kernelization} if there exists a polynomial-time algorithm that, given an instance $(x, k)$ of $\Pi$, outputs an equivalent instance $(x', k')$ of $\Pi$ such that $|x'| + k' \leq g(k)$.
If $g(k)$ is $k^{\cO(1)}$, then we say that $\Pi$ admits a {\em polynomial kernelization}.

\mypara{Parameterized optimization problem and approximate kernels:}
\begin{definition}
\label{defn:para-opt-problem}
A {\em parameterized optimization problem} is a computable function $\Pi: \Sigma^* \times \nn \times \Sigma^* \rightarrow \rr \cup \{\pm \infty\}$.
\end{definition}

The {\em instances} of a parameterized problem are pairs $(x, k) \in \Sigma^* \times \nn$, and a {\em solution} to $(x, k)$ is simply $s \in \Sigma^*$ such that $|s| \leq |x| + k$.
The {\em value} of a solution $s$ is $\Pi(x, k, s)$.
Since the problems we deal with here are minimization problems, we state some of the definitions only in terms of minimization problems (for maximization problems, the definition would be analogous).
As an illustrative example, we provide the definition of the parameterized optimization version of {\CBTDS} problem as follows. This is a minimization problem that is a function ${\CBT}: \Sigma^* \times \nn \times \Sigma^* \rightarrow \rr \cup \{\pm \infty\}$ as follows.

\sloppy
\[   
{\CBT} (G, k, S) = 
     \begin{cases}
       \infty & \text{if } S \text{ is not a {\cbtds} of }G\\
       \min\{|S|, k+1\} &\text{otherwise.} \\ 
     \end{cases}
\] 

\begin{definition}
\label{defn:parameterized-minimization}
For a parameterized minimization problem $\Pi$, the {\em optimum value} of an instance $(x, k)$ is $\OPT_{\Pi} (x, k) = \min_{s \in \Sigma^*, |s| \leq |x| + k} \Pi(x, k, s)$.
\end{definition}

\sloppy
For the case of {\CBTDS}, we define $\OPT(G, k) = \min_{S \subseteq V(G)}\{\CBT(G, k, S)\}$.
We now recall the other relevant definitions regarding approximate kernels.

\begin{definition}
\label{defn:approximate-preprocessing}
Let $\alpha \geq 1$ be a real number and let $\Pi$ be a parameterized minimization problem.
An {\em $\alpha$-approximate polynomial-time preprocessing algorithm} $\cA$ is a pair of polynomial-time algorithms.
The first one is called the {\em reduction algorithm} and the second one is called the {\em solution-lifting algorithm}.
Given an input instance $(x, k)$ of $\Pi$, the reduction algorithm is a function $\cR_{\cA}: \Sigma^* \times \nn \rightarrow \Sigma^* \times \nn$ that outputs an instance $(x', k')$ of $\Pi$.

The solution-lifting algorithm takes the input instance $(x, k)$, the reduced instance $(x', k')$ and a solution $s'$ to the instance $(x', k')$.
The solution-lifting algorithm works in time polynomial in $|x|, k, |x'|, k'$, and $|s'|$, and outputs a solution $s$ to $(x, k)$ such that the following holds: $$\frac{\Pi(x, k, s)}{\OPT_{\Pi}(x, k)} \leq \alpha \frac{\Pi(x', k', s')}{\OPT_{\Pi}(x', k')}$$

\smallskip
\sloppy
The {\em size} of a polynomial-time preprocessing algorithm $\cA$ is a function ${\rm size}_{\cA}:\nn \rightarrow \nn$ defined as ${\rm size}_{\cA}(k) = {\rm sup}\{|x'| + k': (x',k') = \cR_{\cA}(x, k), x \in \Sigma^*\}$.
\end{definition}

\begin{definition}[Approximate Kernelization]
\label{defn:approx-kernelization}
An {\em $\alpha$-approximate kernelization} (or {\em $\alpha$-approximate kernel}) for a parameterized optimization problem $\Pi$, and a real $\alpha \geq 1$ is an $\alpha$-approximate polynomial-time preprocessing algorithm $\cA$ for $\Pi$ such that ${\rm size}_{\cA}$ is upper-bounded by a computable function $g: \nn \rightarrow \nn$.
If $g$ is a polynomial function, we call $\cA$ an {\em $\alpha$-approximate polynomial kernelization algorithm}.
\end{definition}

\begin{definition}[Approximate Kernelization Schemes]
\label{defn:approx-kernel-schemes}
A {\em polynomial-size approximate kernelization scheme (PSAKS)} for a parameterized problem $\Pi$ is a family of $\alpha$-approximate polynomial kernelization algorithms, with one such algorithm for every fixed $\alpha > 1$. 
\end{definition}

\begin{definition}[Time-efficient PSAKS]
\label{defn:time-efficient-PSAKS}
A PSAKS is said to be {\em time efficient} if both the reduction algorithm and the solution lifting algorithms run in $f(\alpha)|x|^{c}$ time for some function $f$ and a constant $c$ independent of $|x|, k$, and $\alpha$.
\end{definition}

\section{Approximate kernel for \textsc{Connected \(\eta\)-Treedepth Deletion}}\label{sec:approxKernel}

In this section, we describe a $(1+\eps)$-approximate kernel for {\CBTDS}. For the entire proof, let us fix a constant \(\eta\in \mathbb{N}\), the instance \((G,k)\) of \textsc{Connected \(\eta\)-Treedepth Deletion}, as well as \(\eps\in \mathbb{R}\) such that \(0<\eps\le 1\).
We prove by Theorem~\ref{thm:psaks-CBTDS-main} that {\CBTDS} admits a $(1+\eps)$-approximate kernel with $k^{\OO (\ceil{{2^{\eta + \ceil{10/\eps}}\eta}/{\eps}})}$ vertices.
As $\eta$ and $\eps$ are fixed constants, the hidden constants in Big-Oh notation could depend both on $\eta$ and $\eps$.

\paragraph{\bf Overview of the Algorithm.}

{Our reduction algorithm works in three phases.

\begin{itemize}
	\item {\bf Phase 1:}
First, observe that in order for a \cbtds\ to exist, at most one connected component of \(G\) can have treedepth more than \(\eta\) and hence, we may focus on the case when \(G\) is connected. We then show that we can decompose the graph \(G\) into three sets \(X\), \(Z\), and \(R\) such that \(X\) is an \btds, the size of the neighborhood of every component \(C\) of \(G[R]\) in \(Z\) is at most \(\eta\) and every \btds\ \(S\) of size at most \(k\) hits all but at most \(\eta\) neighbors of \(C\) in \(X\). {\color{black} We describe this part in Section~\ref{sec:decomposition-of-G}.}

	\item {\bf Phase 2:}
Unfortunately, we cannot identify which of the (at most) \(\eta\) vertices will not be in the solution. However, if the neighborhood of \(C\) is large (at least some constant depending on \(\eta\) and \(\eps\)), then including the whole neighborhood in the solution is not ``too'' suboptimal, as we can add the \(\eta\) vertices to \(S\) and connect them using at most \(2^{\eta}\eta\) additional vertices from within \(C\). While we cannot remove the neighborhood of \(C\) from \(G\) at this point, as we are not able to ensure at this point that it will be connected to the solution we find, we can force the neighborhood of \(C\) in every solution by adding a small gadget to \(G\). Repeating this procedure allows us to identify a set of vertices \(H\subseteq X\) that we can safely force into a solution without increasing the size of an optimal solution too much. Moreover, we obtain that every component of \(G[R]\) has only constantly many neighbors outside of \(H\).
{\color{black} For the description and analysis of this phase, see Reduction Rule~\ref{red-rule:forcing-neighborhood-of-C}, Lemma~\ref{lemma:number-of-executions-rule-2-treedepth}, and Lemma~\ref{lemma:alpha-safeness-red-rule-forcing-neighborhood-of-C}}.

	\item {\bf Phase 3:} 
Notice now that the vertices of any solution for \CBTDS\ can be split into two parts - {\em the obstruction hitting vertices}, such that removal of these vertices guarantees treedepth at most \(\eta\), and {\em the connector vertices} that are only there to provide connectivity to a solution. Now all connected components of \(G[R]\) have treedepth at most \(\eta\). 
Moreover, we are guaranteed that any solution of size at most \(k\) contains all but at most \(2\eta\) neighbors of a connected component \(C\) of \(G[R]\). Hence, if our goal was only hitting the obstructions in \(G\), then we could assume that \(S\) contains at most \(2\eta\) vertices of every connected component of \(G[R]\). However, there are two problems. We do not know which \(2\eta\) vertices are in \(N(C)\setminus S\) and vertices of \(C\) can also provide connectivity to \(S\). This requires the use of careful problem-specific argumentation and reduction rules. 

To resolve the first problem, we observe that \(N(C)\setminus H\) has already constant size and we can classify the (subsets of) vertices of \(C\) into types depending on their neighborhood in \(N(C)\setminus H\). To resolve the second issue, we allow each connected component of \(G[R]\) to have much larger, albeit still constant, intersection with \(S\).
We furthermore observe that if we chose this constant, denoted by \(\lambda\), then we can include the whole neighborhood of every component \(C\) that intersects a solution \(S\) in more than \(\lambda\) vertices, into the solution without increasing the size of the solution too much.
Now, we finally can identify vertices that are not necessary for any solution that intersects every component in at most \(\lambda\) vertices. 
 Denote the set of these vertices \(\cM\). There is no danger in removing such vertices for hitting the obstructions, as for every component \(C\) of \(G'-(X\cup Z)\) that intersects more than \(\lambda\) vertices of a solution in the reduced instance \(G'\) our solution-lifting algorithm adds the neighborhood of \(C\) into the solution. However, removing all of these vertices may very well destroy the connectivity of the solution. Here, we make use of Propositions~\ref{proposition:approx-k-steiner-tree}~and~\ref{proposition:steiner-tree-terminals} to find a small subset \(\cN\) of vertices in \(\cM\) such that \(G- (\cM\setminus \cN)\) actually have a \cbtds\ of approximately optimal size.
 {\color{black} The description of this last and most crucial phase is given in Section~\ref{sec:good-solution-characteristics} and Section~\ref{sec:removing-irrelevant-vertices}}.
\end{itemize}
 
\subsection{Decomposition of the Graph \(G\)}
\label{sec:decomposition-of-G}

We first observe that we can remove all connected components of \(G\) that already have treedepth at most \(\eta\), as we do not need to remove any vertex from such a component.
\begin{reduction rule}
\label{red-rule:redundant-component-removal}
Let $C$ be a connected component of $G$ such that $\td (G[C]) \leq \eta$.
Then, delete $C$ from $G$.
The new instance is $(G - C, k)$.
\end{reduction rule}

Observe that Reduction Rule~\ref{red-rule:redundant-component-removal} is an approximation preserving reduction rule.
What it means is that given a $c$-approximate {\cbtds} of $(G - C, k)$, we can in polynomial time compute a $c$-approximate {\cbtds} of $(G, k)$. 
Suppose that $G$ has two distinct connected components $C_1$ and $C_2$ such that ${\td}(G[C_1]), \td (G[C_2]) > \eta$. 
Then any {\btds} for $G$ has to contain vertices from both connected components of $G$.
It implies that $G$ does not admit any {\cbtds}.
Therefore we can output any constant size instance without a solution, such as $(K_{\eta + 2} \uplus K_{\eta + 2}, 1)$ as the output instance. 
Hence, we can assume that $G$ is a connected graph.

We start by constructing a decomposition of the graph such that $V(G) = X \uplus Z \uplus R$ satisfying some crucial properties that we use in our subsequent phases of the preprocessing algorithm.
The construction is inspired by the decompositions used by Fomin et al.~\cite{FominLMS12} (for the {\sc Planar $\cF$-Deletion} problem) and Giannopoulou et al.~\cite{GiannopoulouJLS17} (for the specific case of {\sc $\eta$-Treedepth Deletion}).

\begin{lemma}
	\label{lemma:decomposition-of-G}
	There exists a polynomial-time algorithm that either correctly concludes that no \btds\ $S$ for \(G\) of size at most $k$ exists, or it constructs a partition $V(G) = X \uplus Z \uplus R$ such that the following properties are satisfied.
	\begin{enumerate}
		\item\label{decom-prop-1} $X$ is an $\eta$-treedepth deletion set of $G$ and $|X| = \cO(k)$.
		\item\label{decom-prop-2} $|Z| = \cO(k^{3})$.
		\item\label{decom-prop-3} For every connected component $C$ of $G[R]$,{$|N_G(C) \cap Z| \leq \eta$.}
		\item\label{decom-prop-4} Let $C$ be a connected component of $G[R]$. Then, for any $\eta$-treedepth deletion set $S$ of size at most $k$, it holds that $|(N_G(C) \cap X) \setminus S| \leq \eta$. 
	\end{enumerate}
\end{lemma}

\begin{proof}
	We start by computing an \(\cO(1)\)-approximation for \(\eta\)-\textsc{Treedepth Deletion Set} given by Proposition~\ref{prop:approximation}. Let \(X\) be the output of the approximation algorithm. Note that \(X\) is  an $\eta$-treedepth deletion set of $G$. Since the algorithm of Proposition~\ref{prop:approximation} is \(\cO(1)\)-approximation, we get that either $|X| = \cO(k)$, or we can correctly conclude that $G$ has no $\eta$-treedepth deletion set of size at most $k$. For a connected component \(C\) in \(G-X\), let \(Y_C\) denotes some treedepth decomposition of \(G[C]\) of depth \(\eta\) and let \(Y = \bigcup_{C \text{ is a component of } G-X} Y_C\). Observe that \(Y\) is a treedepth decomposition of \(G-X\). As \(\eta\) is a constant, we can compute an \(Y_C\) for each connected component \(C\) in linear time invoking the algorithm from Proposition~\ref{proposition:computing_elimination_tree}. We will construct \(Z\) as follows. We start with \(Z=\emptyset\). Now for every pair of non-adjacent vertices \(x,y\in X\), we compute in polynomial time a minimum (vertex) \(x\)-\(y\) cut \(Q_{x,y}\) in the graph \(G-X\)~\cite{ford1962flows}. If the size of \(Q_{x,y}\) is at most \(k+\eta\), then we add \(Q_{x,y}\) together with its upward closure in \(Y\) to \(Z\). Note that each rooted tree in \(Y\) has depth at most \(\eta\).
	Hence \(Q_{x,y}\) together with its upward closure in \(Y\) has at most \(\eta|Q_{x,y}|\) vertices and for each pair of vertices \(x,y\) in \(X\) we add to \(Z\) at most \(\eta k + \eta^2\) many vertices. It follows that \(|Z| = \cO(k^3)\). Finally we let \(R = V(G)\setminus(X\cup Z)\). It remains to show that Properties~\ref{decom-prop-3}~and~\ref{decom-prop-4} are satisfied. 
	
	Towards Property~\ref{decom-prop-3}, let \(C\) be a connected component of \(G[R]\) and let \(C'\) be the connected component of \(G-X\) such that \(C\subseteq C'\). We claim that \(N_G(C) \cap Z\) is a subset of the common ancestors of all the vertices in \(C\) in the treedepth decomposition \(Y_{C'}\). Now let \(p\in C\) and \(z\in N_G(C)\cap Z\) be such vertices that \(z\) is not an ancestor of \(p\) and let \(q\in C\) be a neighbor of \(z\). Note that \(N(C)\subseteq C'\cup X\), so \(z\in C'\). As \(Y_{C'}\) is a treedepth decomposition for \(G[C']\) and \(qz\) is an edge in \(G[C']\), it follows that \(q\) is either ancestor or descendant of \(z\). However, all ancestors of \(z\) in \(Y_{C'}\) are in \(Z\). Hence \(q\) is a descendant of \(z\). Now if \(p\) is a descendant of \(q\), then it is also a descendant of \(z\). Similarly, if \(p\) is an ancestor of \(q\), then either $p$ is an ancestor or $q$ or a descendant of $q$. If $q$ is a descendant of \(z\), then this leads to a contradiction to our assumption.
	Otherwise $q$ is an ancestor in \(z\). However, every ancestor of \(z\) is in \(Z\), which contradicts the fact that \(p\) is in \(C\). It follows that \(p\) and \(q\) are not in an ancestor-descendant relation. Moreover, their least common ancestor, denoted \(\operatorname{lca}(p,q)\), in \(Y_{C'}\) is an ancestor of \(z\), as all the ancestors of their least common ancestor are also their common ancestors and \(z\) is not an ancestor of \(p\). The fact that \(\operatorname{lca}(p,q)\) is an ancestor of \(z\) then follows from the fact that both \(z\) and \(\operatorname{lca}(p,q)\) are ancestors of \(q\). Hence all common ancestors of \(p\) and \(q\) in \(Y_{C'}\) are ancestors of \(z\) and therefore in \(Z\). However, it is well-known and easy to see that the common ancestors of two vertices (that are not in the ancestor-descendant relationship) in a treedepth decomposition of a graph form a vertex cut between these two vertices. It follows that \(p\) and \(q\) cannot be in the same connected component of \(G[R]\), a contradiction. Therefore, indeed every vertex in \(N(C)\cap Z\) is an ancestors of all vertices in \(C\) in the treedepth decomposition \(Y_{C'}\). However, the height of the treedepth decomposition \(Y_{C'}\) is at most \(\eta\), therefore \(|N_G(C) \cap Z| \leq \eta\).
	
	Towards Property~\ref{decom-prop-4}, let \(C\) be a connected component of \(G[R]\), \(S\) be an \btds\ of size at most $k$, and \(Y_S\) be a treedepth decomposition of \(G-S\) such that the depth of each rooted tree in \(Y_S\) is at most \(\eta\). We will show that for every pair of vertices \(x,y\in (N_G(C) \cap X) \setminus S\) it holds that \(x\) and \(y\) are in the ancestor-descendant relationship in \(Y_S\). Note that this is possible if and only if all vertices in \((N_G(C) \cap X) \setminus S\) are on a single leaf-to-root path in some tree of \(Y_S\) and Property~\ref{decom-prop-4} follows from the fact that the depth of such tree is at most \(\eta\). Now, if there is an edge \(xy\in E(G)\), then \(x\) and \(y\) are in the ancestor-descendant relationship in \(Y_S\) by the definition of a treedepth decomposition for a connected graph. So we can assume that \(x\) and \(y\) are not adjacent in \(G\). Since \(C\) is a connected component of \(G[R]\) and \(x,y\in (N_G(C) \cap X)\), it follows that there exists a path from \(x\) to \(y\) in \(G[R\cup \{x,y\}]\). Moreover, if the size of the minimum \(x\)-\(y\) cut from \(x\) to \(y\) is at most \(k+\eta\), then \(Z\) contains one such  minimum \(x\)-\(y\) cut and hence there is no path from \(x\) to \(y\) in \(G[R\cup \{x,y\}]\). Therefore, the size of minimum \(x\)-\(y\) cut in \(G\) is at least \(k+\eta+1\). On the other hand, if \(x\) and \(y\) are not in the ancestor-descendant relationship in \(Y_S\), then the set \(\operatorname{CA}_{x,y}\) of their common ancestors forms a \(x\)-\(y\) cut in \(G-S\) and consecutively \(S\cup \operatorname{CA}_{x,y} \) is an \(x\)-\(y\) cut of size at most \(k+\eta\), a contradiction.
\end{proof}

{We run the algorithm of Lemma~\ref{lemma:decomposition-of-G} and we fix for the rest of the proof the sets of vertices \(X\), \(Z\), and \(R\) such that they satisfy the above lemma. Furthermore,} let us fix a $\delta=\frac{\eps}{10}$ and notice that since \(\eps\le 1\), we have that $(1+\delta)^4 \leq (1+\eps)$. Finally let us set $d = \ceil{\frac{{2^{\eta +3} \eta}}{\delta}}$.
{The next step of the algorithm is to find a set of vertices \(H\subseteq X\) such that every component \(C\) of \(G-(X\cup Z)\) has at most \(d+\eta\) neighbors in \(X\setminus H\). Our goal is to do it in a way that we can force \(H\) into every solution and increase the size of an optimal solution only by a small fraction.} 

\subsection{Processing Connected Components of $G - (X \cup Z)$ with Large Neighborhoods}
\label{subsec:large-component-treedepth}

We initialize $H := \emptyset$ and we apply the following reduction rule exhaustively.

\begin{reduction rule}
\label{red-rule:forcing-neighborhood-of-C}
Let $C$ be a connected component of $G[R]$. If {$|(N_G(C)\cap X)  \setminus H| > d + \eta$},
then for every $u \in N_G(C)\cap X$, add a {new} clique $J$ with $\eta + 1$ vertices {to \(G\)} such that $J \cap X = \{u\}$ and {\(N_G(J\setminus\{u\})=\{u\}\)}.
Add the vertices of $N_G(C)$ to $H$.
\end{reduction rule}

After we finish applying Reduction Rule~\ref{red-rule:forcing-neighborhood-of-C} on $(G, k)$ exhaustively, let {$G'$} be the resulting graph.
We prove the following two lemmas using Lemma \ref{lemma:decomposition-of-G}.

\begin{lemma}
\label{lemma:number-of-executions-rule-2-treedepth}
Let $S$ be an optimal {\cbtds} of $(G, k)$ of size at most $k$.
Then, Reduction Rule \ref{red-rule:forcing-neighborhood-of-C} is not applicable more than $|S|/d$ times.
\end{lemma}

\begin{proof}
Let $S$ be an optimal {\cbtds} of $(G, k)$ of size at most $k$.
Due to the item
(\ref{decom-prop-4}) of Lemma~\ref{lemma:decomposition-of-G}, $|(N(C) \cap X) \setminus S| \leq \eta$.
By the precondition(s), Reduction Rule~\ref{red-rule:forcing-neighborhood-of-C} is applicable only when $|(N(C)\cap X) \setminus H| > d + \eta$.
But, $S$ must contain at least $d$ vertices from $N(C) \setminus H$.
So, one execution of Reduction Rule~\ref{red-rule:forcing-neighborhood-of-C} {adds} at least $d$ new vertices from {$S$ to \(H\) and the lemma follows.}
\end{proof}

Using the above lemma, we prove the following lemma.

\begin{lemma}
\label{lemma:alpha-safeness-red-rule-forcing-neighborhood-of-C}
Let $(G',k')$ be the instance obtained after exhaustively applying the Reduction Rule~\ref{red-rule:forcing-neighborhood-of-C} on $(G, k)$ such that $k' = k$.
Then, the following conditions are satisfied.
\begin{itemize}
	\item Any {\cbtds} of $(G', k')$ is a {\cbtds} of $(G, k)$, and
	\item If $\OPT(G, k) \leq k$, then $\OPT(G', k') \leq (1+\delta)\OPT(G, k)$
\end{itemize}
\end{lemma}

\begin{proof}
Let us prove the statements in the given order.
	
	\begin{claim}
		Any {\cbtds} of $(G', k')$ is a {\cbtds} of $(G, k)$.
	\end{claim}
	
\begin{claimproof}
	{Let \(S'\) be a \cbtds\ of \((G',k')\). Note that \(G\) is an induced subgraph of \(G'\), hence \(S=S'\cap V(G)\) is an \(\eta\)-treedepth deletion set of \((G,k)\). It only remains to show that \(S\) is connected. If \(S'\) does not contain a vertex from any clique \(J\) (added by some execution of Reduction Rule~\ref{red-rule:forcing-neighborhood-of-C}) other than the single vertex in \(J\cap H\), then \(S=S'\) and it is connected. Else, let \(J\) be a clique added by some application of Reduction Rule~\ref{red-rule:forcing-neighborhood-of-C} such that \(S'\cap (J\setminus H)\) is not empty. Note that \(N_G(J\setminus H)\) is a singleton containing some vertex \(u\in X\) and if \(S'\setminus J\) is non-empty then \(S'\) has to contain \(u\). It is easy to see that \(S'\setminus (J\setminus H)\) is also a \cbtds. The first item follows by repeating the same argument for every clique added by some application of Reduction Rule~\ref{red-rule:forcing-neighborhood-of-C}.}
\end{claimproof}
	\begin{claim}
		If $\OPT(G, k) \leq k$, then $\OPT(G', k') \leq (1+\delta)\OPT(G, k)$.
	\end{claim}
	\begin{claimproof}
			Let $S$ be an optimal solution of $(G, k)$ and $|S| \leq k$.
		We construct a feasible solution $S'$ of $(G', k')$ as follows.
		We set $S' := S \cup H$ first.
		{Note that \(S'\) is an \(\eta\)-treedepth deletion set of \(G'\), because for every newly added clique \(J\), we have that \(J\setminus S'\) is a connected component of \(G-S'\) of size \(\eta\) and therefore also treedepth of this component is at most \(\eta\).}
		However, this does not guarantee that $G[S \cup H]$ is connected.
		
		{Let \(h\) be a vertex in \(H\setminus S\) and let \(C\) be the component of \(G[R]\) such that the application of Reduction Rule~\ref{red-rule:forcing-neighborhood-of-C} on the component \(C\) added \(h\) to \(H\). Note that { \(|N_G(C)\cap X|> d+\eta\)} and since $|S| \leq k$, it follows from Lemma~\ref{lemma:decomposition-of-G} that \(|(N_G(C)\cap X)\setminus S|\le \eta\). Consecutively \(N_G(C)\cap X\cap S\) is not empty. Moreover, \(G[C]\) has treedepth at most \(\eta\) and by Proposition~\ref{lemma:bounded-treedepth-diameter} diameter at most \(2^{\eta}\). Hence the shortest path from \(h\) to any vertex of \(S\) has length at most \(2^{\eta}\). Applying the above argument for all vertices of \(H\setminus S\) we get that}
		at most $2^\eta|H \setminus S|$ additional vertices are required to make $H \cup S$ connected.
		We add those vertices and update the set $S'$.
		It implies that $|S'| \leq |S| + |H \setminus S|2^\eta$.
		Let us argue that $|S'| \leq (1+\delta)|S|$, that is \(|H \setminus S|2^\eta\le \delta|S|\).	
		{By} Lemma~\ref{lemma:number-of-executions-rule-2-treedepth}, Reduction Rule~\ref{red-rule:forcing-neighborhood-of-C} is applicable at most $|S|/d$ times.
		So, we can partition {$H \setminus S = H_1 \uplus H_2 \uplus \cdots \uplus H_{\ell}$ such that \(\ell \le \frac{|S|}{d}\)} and $H_i$ is the set of vertices added to $H$ at the $i$'th execution of Reduction Rule~\ref{red-rule:forcing-neighborhood-of-C}.
		But for all {$i \in [\ell]$}, $H_i$ is the set of vertices that are {in \(N(C)\cap X\) for some component \(C\) in \(G[R]\) and} outside $S$.
		So, {by Lemma~\ref{lemma:decomposition-of-G},  $|H_i| \leq \eta$}.
		{Hence, \(|H \setminus S|2^\eta \le \eta\cdot \ell\cdot 2^{\eta} \le  \eta\cdot \frac{|S|}{d}\cdot 2^{\eta} = \frac{\eta 2^{\eta}}{d}|S|\le \delta |S|\).}
	\end{claimproof}
This completes the proof of the lemma.
\end{proof}

\subsection{Understanding the structure of a good solution}
\label{sec:good-solution-characteristics}

{From now on we assume that we have applied  Reduction Rule~\ref{red-rule:forcing-neighborhood-of-C} exhaustively and, for the sake of exposition, we denote by \(G\) the resulting graph. Moreover, we also fix the set \(H\) we obtained from the exhaustive application of Reduction Rule~\ref{red-rule:forcing-neighborhood-of-C}.
It follows that every connected component of $G - (X \cup Z)$ have at most $d + 2\eta$ neighbors outside $H$.}

Furthermore, Reduction Rule~\ref{red-rule:forcing-neighborhood-of-C} ensures the following observation.

\begin{observation}
\label{obs:red-rule-2-implication}
Any feasible {\cbtds} of $(G, k)$ must contain $H$.
\end{observation}

The above observation follows because for every vertex \(u\in H\) there exists a clique \(J\) of size \(\eta+1\) that contains \(u\) and \(N_G(J\setminus \{u\})=\{u\}\).
So every {\cbtds} that contains a vertex in \(J\) and a vertex outside of \(J\) contains also \(u\). Notice that if \(S\) is a \cbtds\ for \((G,k)\) and \(C\) is a component of \(G[R]\), then \((S\setminus C)\cup N(C)\) is an \btds.
Moreover by Proposition \ref{lemma:bounded-treedepth-diameter}, we can connect each vertex from \(N(C)\setminus H\) to \(H\) using at most \(2^{\eta}\) vertices of \(C\). 
Hence the only reason for a component of \(G[R]\) to contain more than \(2^{\eta}(d+2\eta)\) vertices is if \(C\) also provides connectivity to \(S\). Let us fix for the rest of the proof $\lambda = 2^{\eta}\ceil{\frac{d + 2\eta}{\delta}}$. 

Let $T \subseteq (X\cup Z)\setminus H$. We denote by $\components{T}$ the set of all the components $C$ of $G - (X \cup Z)$ such that $N(C)\setminus H = T$. Note that, by the definition of $H$, if $|T|\ge d+2\eta+1$, then $\components{T}= \emptyset$.
Let $S$ be an $\eta$-treedepth deletion set of $G$.
Suppose that for every $T\subseteq (X\cup Z)\setminus H$, {if $S$ intersects $\components{T}$ in more than $\lambda$ vertices (i.e. $|\bigcup_{C\in \components{T}}(S\cap C)| > \lambda$), then $T\subseteq S$.
Then we say that $S$ is a \emph{\nice} treedepth deletion set of $G$.

From now on, we focus on \nice\ \cbtds{s}. We first reduce the instance \((G,k)\) to an instance \((G',k)\) such that 
\begin{itemize}
	\item \(G'\) is an induced subgraph of \(G\),
	\item every nice {\cbtds} for \((G',k)\) is also a \nice\ \cbtds\ for \((G,k)\), and
	\item \((G',k)\) has a \nice\ \cbtds\ of size at most \((1+\delta)^2\OPT(G,k)\).
\end{itemize}

Afterwards, we show that any \cbtds\ \(S'\) for \((G',k)\) can be transformed into a \nice\ \cbtds\ for \((G',k)\) of size at most \((1+\delta)|S'|\).
To obtain our reduced instance we will heavily rely on the following lemma that helps us identify the vertices that only serve as connectors in any \nice\ \cbtds.}

\begin{lemma}
\label{lemma:component-characteristics}
	Let \(G'\) be an induced (not necessarily strict) subgraph of \(G\) and \(T, C_1, C_2,\ldots, C_\ell\) be pairwise disjoint sets of vertices in \(G\) such that:
	\begin{itemize}
		\item \(G'[C_i]\) is connected,
		\item \(N(C_i)\setminus H = T\), for some fixed set of vertices \(T\), and
		\item \(\td(G'[C_i]) = \td(G'[C_j])\) for all \(i,j\in [\ell]\).
	\end{itemize}	
	Now let \(S\) be an \(\eta\)-treedepth deletion set in \(G'\) such that \(H\subseteq S\) and let \(\mathcal{J} = \{C_i\mid C_i\cap S = \emptyset \}\), i.e., \(\mathcal{J}\) is the set of components in  \(C_1, C_2,\ldots, C_\ell\) that do not contain any vertex of \(S\). 	
	If \(|\mathcal{J}|\ge \eta + 1\), then \(S' = S\setminus (\bigcup_{i\in[\ell]} C_i)\) is an \(\eta\)-treedepth deletion set in \(G'\). 
\end{lemma}

\begin{proof}
	First note that \(\td(G'[C_i])\le \eta\), otherwise for every \(C_i\in \mathcal{J}\) we have \(G'[C_i]\subseteq G'-S\) and \(\td(G'[C_i])> \eta\). Since treedepth is closed under taking subgraphs, this implies \(\td(G'-S)> \eta\). Moreover, if \(T\subseteq S\), then for all \(i\in [\ell]\) the vertex set \(C_i\) induces a connected component of \(G'-S\). Since \(C_i\cap T=\emptyset\), we get that for all \(i\in [\ell]\) the set \(C_i\) induces a connected component of \(G'-S'\). Hence, in this case we get that every \(C_i\) is in its own component of \(G'-S'\) of treedepth at most \(\eta\). We can now assume that \(T\setminus S\) is not empty. 
	
	Let \(F\) be a treedepth decomposition of \(G-S\) of depth \(\eta\) with root \(r\). For \(C_i\in \mathcal{J}\), let \(x_i\) be the vertex of \(C_i\) with the minimum distance to \(r\) in \(F\). Note that since \(G'[C_i]\) is connected, it follows from the properties of the treedepth decomposition that all vertices of \(C_i\) are in the subtree of \(F\) rooted in \(x_i\). Hence the depth of the tree rooted in \(x_i\) is at least \(\td(G'[C_i])\). 
	Now for every vertex \(w\in T\setminus S\), there is an edge between \(w\) and some vertex \(y_i\) in \(C_i\) and \(w\) is either an ancestor or descendant of \(y_i\) in \(F\). Moreover, \(x_i\) is an ancestor of \(y_i\) and consecutively \(w\) is either an ancestor or a descendant of \(x_i\). Now if \(w\) is a descendant of \(x_i\), then for all \(C_j\in \mathcal{J}\setminus\{C_i\}\), we get that either \(x_j\) is a descendant of \(w\) and hence of \(x_i\) as well or \(x_j\) is an ancestor of \(w\). In this case both \(x_i\) and \(x_j\) are ancestors of \(w\) and they are on the unique path from \(w\) to \(r\). Hence, either \(x_i\) is an ancestor of \(x_j\) or \(x_j\) is an ancestor of \(x_i\). It follows that if for every component \(C_i\in \mathcal{J}\) there is a vertex \(w\in T\setminus S\) such that \(w\) is a descendant of \(x_i\), then for every pair of components \(C_i, C_j\in \mathcal{J}\), we get that the vertices \(x_i, x_j\) are in the ancestor-descendant relationship.
	This is only possible if all vertices \(x_i\), \(C_i\in \mathcal{J}\), are on a single leaf-root path.  But \(|\mathcal{J}|\ge \eta + 1\) and we have a contradiction with the fact that \(F\) is a treedepth decomposition of \(G-S\) of depth at most \(\eta\). Therefore, for at least one \(C_i\in \mathcal{J}\) we have that the vertex \(x_i\) is descendant of all vertices in \(T\setminus S\) in the treedepth decomposition \(F\). To get a treedepth decomposition \(F'\) of \(G'-S'\) we simply add \(\ell-|\cJ|\) children in \(F\) to the parent of \(x_i\) and attach a treedepth decomposition of height \(\td(G'[C_j])=\td(G'[C_i])\) for each \(C_j\notin \mathcal{J}\) to one of the newly added vertices of \(F\). Since \(x_i\) is at distance at most \(\eta-\td(G'[C_i])\) from the root \(r\), it follows that the height of the resulting rooted tree \(F'\) is at most \(\eta\). Furthermore, all vertices in \(T\setminus S\) are ancestors of \(x_i\) and hence are ancestors of all the newly added vertices. Since the neighborhood of every vertex in \(C_j\) in the graph \(G'-S'\) is a subset of \(C_j\cup (T\setminus S)\), it follows that \(F'\) is indeed a treedepth decomposition of \(G'-S'\). 
\end{proof}



\subsection{Identifying Further Irrelevant Vertices}
\label{sec:removing-irrelevant-vertices}

{We now mark some vertices of $G - (X \cup Z)$ that we would like to remove, as these vertices are not important for hitting obstructions in a \nice\ \cbtds. We note that we will end up not removing all of these vertices, as some of them will be important as connectors for obstruction hitting vertices in the solution. However, this step lets us identify a relatively small subset of vertices such that any \nice\ \btds\ for the subgraph induced by this subset of vertices is indeed \nice\ \btds\ for \(G\). We then make use of Propositions~\ref{proposition:approx-k-steiner-tree}~and~\ref{proposition:steiner-tree-terminals} to add some vertices back as possible connectors.}
{Recall that we fixed  $\lambda = 2^{\eta}\ceil{\frac{d + 2\eta}{\delta}}$. 
Let us set \(\cM=\emptyset\). We now describe two reduction rules based on Lemma~\ref{lemma:component-characteristics} that do not change \(G\) and only add vertices to \(\cM\).}
 {For \(T\subseteq (X\cup Z)\setminus H\) and \(i\in \mathbb{N}\), let \(\componentsTD{T}{i}\) denote the components \(C\in \components{T}\) such that \(\td(G[C])=i\).
} 
 \begin{reduction rule}\label{rrule:components}
 	Let $T \subseteq (X \cup Z) \setminus H$ and $i \in [\eta]$. If \(|\componentsTD{T}{i}|\ge \lambda + \eta + 2\), then add vertices of all but \(\lambda + \eta + 1\) of the components in \(\componentsTD{T}{i}\) to \(\cM\).
 \end{reduction rule}

\begin{reduction rule}
\label{rrule:single_component}
Let $C$ be a component of $G- (X \cup Z \cup \cM)$, $Y_C$ be a treedepth decomposition of $G[C]$ of depth at most $\eta$, and $i \in [\eta]$.
Moreover, let $v$ be a vertex in $C$ and $T\subseteq (N(C)\setminus H)\cup \uclos_{Y_C}(\{v\})$.
Finally, let $\mathcal{C} =\{C_1, C_2, \ldots, C_\ell\}$ be all the components of $G - T$ such that for all $j\in [\ell]$ it holds that $C_j\subseteq C$, $N(C_j)\setminus H = T$ and $\td (G[C_j])=i$.
If $|\mathcal{C}| \geq \lambda + \eta + 2$, then add the vertices of all but $\lambda + \eta + 1$ of the components in $\mathcal{C}$ to $\cM$.
\end{reduction rule}

Once we apply Reduction Rules \ref{rrule:components} and \ref{rrule:single_component} exhaustively and obtain a vertex set $\cM$, we use Lemma \ref{lemma:component-characteristics} in order to to prove the following two lemmas that provide some interesting characteristics (the following two lemmas) of nice {\btds}{s} of $G$. 

\begin{figure}[t]
\centering
	\includegraphics[scale=0.2]{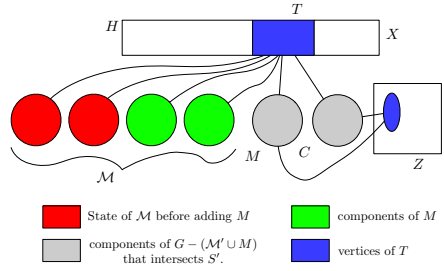}
	\caption{Illustration of Lemma \ref{lemma:component-characteristics} when Reduction Rule \ref{rrule:components} is applied. All components colored red, green, and grey are of equal treedepth and have exact neighborhood outside $H$.}
\label{fig:nice-solution-illustration}	
\end{figure}

\begin{lemma}
\label{lemma:nice_solution_after_marking}
	Let \(\cM\) be the set of vertices obtained by exhaustive application of Reduction Rules~\ref{rrule:components}~and~\ref{rrule:single_component} and let \(S'\) be a \nice\ \btds\ for \(G-\cM\). Then \(S'\) is a \nice\ \btds\ for \(G\).
(see Figure \ref{fig:nice-solution-illustration} for an illustration for Reduction Rule \ref{rrule:components}.)
\end{lemma}


\begin{proof}
	We prove that the lemma holds after a single application of each of the reduction rules. The lemma then follows by repeating the same argument for each application of a reduction rule. 
	\begin{claim}
		Let \(M\) be the set of vertices added to \(\cM\) by applying Reduction Rule~\ref{rrule:components} for $T \subseteq (X \cup Z) \setminus H$ and $i \in [\eta]$.
		Also suppose that \(\cM'\) is the state of \(\cM\) before adding \(M\). If \(S'\) is a \nice\ \btds\ for \(G-(\cM'\cup M)\), then \(S'\) is a \nice\ \btds\ for \(G-\cM'\).
	\end{claim}
	
	\begin{claimproof}
		Note that \(M\) is a union of some connected components in \(G-(X\cup Z)\), each with treedepth \(i\le \eta\), and hence \(\td(G[M])\le \eta\). Hence, if \(T\subseteq S'\), then the claim follows. Therefore, we can assume that \(T\not\subseteq S'\).
		Since \(S'\) is \nice\ \btds\ for \(G-(\cM'\cup M)\), it follows that \(|\bigcup_{C\in \components{T}}(S'\cap C)| \le \lambda\). Because Reduction Rule~\ref{rrule:components} was applied, there are at least \(\lambda+\eta+1\) components in \(G-(\cM'\cup M)\) of treedepth exactly $i$. Moreover, if $C$ is one of those at least $\lambda + \eta + 1$ components with ${\td}(G[C]) = i$, then \(N(C)\setminus H = T\).
		Therefore, at least \(\eta+1\) of these components do not contain any vertex of \(S'\). All of these components have same neighborhood outside $H$ and treedepth as the components in \(M\).
		As $S'$ is \nice\, and $S' \cup M$ is a {\btds} for $G - {\cM'}$, the Lemma~\ref{lemma:component-characteristics} 
		tells us that we can remove \(M\) from \(S'\cup M\) and preserve the fact that we have a {\btds} of $G - \cM'$.
	\end{claimproof}

\begin{claim}
	Let \(M\) be the set of vertices added to \(\cM\) by applying Reduction Rule~\ref{rrule:single_component} for a component \(C\), a treedepth decomposition \(Y_C\), a vertex \(v\), a set  \(T\subseteq (N(C)\setminus H)\cup \uclos_{Y_C}(\{v\})\), and a set of components \(\mathcal{C} =\{C_1, C_2, \ldots, C_\ell\}\). Furthermore, let us assume that \(\cM'\) is the state of \(\cM\) before adding \(M\). If \(S'\) is a \nice\ \btds\ for \(G-(\cM'\cup M)\). Then \(S'\) is a \nice\ \btds\ for \(G-\cM'\). 
\end{claim}

\begin{claimproof}
{The proof is similar to the proof of the previous claim.}
If \(N(C)\subseteq S\), then the claim follows from the fact that \(\td(G[C])\le \eta\), as \(C\) be a component of \(G-(X\cup Z\cup \cM)\). Hence, we can assume that \(N(C)\not\subseteq S\). In particular, \(N(C)\setminus H\not\subseteq S\) as \(H\) is in every \cbtds. Because \(S\) is a \nice\ \btds, it follows that \(|S\cap C|\le \lambda\). Because Reduction Rule~\ref{rrule:single_component} was applied, there are at least \(\lambda+\eta+1\) components \(C\) in \(G-(\cM'\cup M)\) with \(\td(G[C])=i\) and \(N(C)\setminus H = T\) and at least \(\eta+1\) of them do not intersect \(S\). Moreover, for each connected component \(D\) in \(G[M]\) it holds that  \(\td(G[D])=i\) and \(N(D)\setminus H = T\). Therefore, the Lemma~\ref{lemma:component-characteristics} implies that \(S'\) is a \nice\ \btds\ for \(G-\cM'\).
\end{claimproof}
\noindent
The above two claims complete the proof of the lemma.
\end{proof}

\begin{lemma}\label{lemma:nice_solution_after_marking_size}
	Let \(\cM\) be the set of vertices obtained by exhaustive application of Reduction Rules~\ref{rrule:components}~and~\ref{rrule:single_component}. 
	Then \(|V(G)\setminus \cM|= \bigoh(k^{3d+6\eta})\).
\end{lemma}

\begin{proof}
	Note that for every connected component \(C\) of \(G-(X\cup Z\cup \cM)\) we have that \(|N(C)_{G-\cM}\setminus H|\le d + 2\eta\). Moreover, \(|X\cup Z|=\bigoh(k^3)\). Hence there are at most \(\bigoh(k^{3(d+2\eta)})\) sets \(T\) such that \(T\subseteq (X\cup Z)\setminus H\) and \(\components{T}\neq \emptyset\). For a fixed \(T\), each connected component in \(\components{T}\) has treedepth at most \(\eta\). Therefore, if \(|\components{T}|\ge \eta(\lambda+\eta+1)+1\), then there is \(i\in [\eta]\) such that \(|\componentsTD{T}{i}|\ge \lambda+\eta+2\) and Reduction Rule~\ref{rrule:components} can be applied. Therefore,  \(|\components{T}|\le \eta(\lambda+\eta+1)\) and there are at most \(\bigoh( \eta(\lambda+\eta+1)k^{3(d+2\eta)}) = \bigoh(k^{3(d+2\eta)}) \) many connected components in \(G-(X\cup Z\cup \cM)\). It remains to show that each such component has a constant size. Let \(C\) be a connected component of \(G-(X\cup Z\cup \cM)\) and let \(Y_C\) be a treedepth decomposition for \(G[C]\) of depth at most \(\eta\). If all vertices of \(Y_C\) have at most \((2^{d+3\eta}\eta(\lambda+\eta+1))\) children, then the size of \(C\) is at most \((2^{d+3\eta}\eta(\lambda+\eta+1))^\eta\)  and the lemma follows. Otherwise, there is a vertex \(v\) in \(Y_C\) with at least \((2^{d+3\eta}\eta(\lambda+\eta+1))+1\) children. Now let \(v_1\) and \(v_2\) be two children of \(v\) and let \(u_1\) be any vertex in the subtree of \(Y_C\) rooted in \(v_1\) and \(u_2\) be any vertex in the subtree of \(Y_C\) rooted in \(v_2\). It is easy to see that \(v\) is the least common ancestor of \(u_1\) and \(u_2\) and it follows from the properties of a treedepth decomposition for a graph that \(N(C)\cup \uclos_{Y_C}\{v\}\) is a vertex separator between \(u_1\) and \(u_2\). Consecutively, there are at least \((2^{d+3\eta}\eta(\lambda+\eta+1))+1\) components \(D\) in \(G-(N(C)\cup \uclos_{Y_C}(\{v\})\) with \(D\subseteq C\). Now
	\(|N(C)_{G-\cM}\setminus H|\le d+2\eta\) and \(|\uclos_{Y_C}(\{v\})|\le \eta\) and it follows that there is \(T\subseteq (N(C)\setminus H)\cup \uclos_{Y_C}(\{v\})\) such that for at least \(\eta(\lambda+\eta+1)+1\) of these components we have \(N(D)\setminus H = T\). Moreover, each of these components have treedepth at most \(\eta\) and hence there are at least \(\lambda+\eta+2\) components with the same neighborhood and treedepth and Reduction Rule~\ref{rrule:single_component} can be applied. Since this is not possible, we conclude that the size of \(C\) is at most \((2^{d+3\eta}\eta(\lambda+\eta+1))^\eta\), i.e., a constant depending only on \(\delta\) and \(\eta\) and the lemma follows.
\end{proof}

Now our next goal is to add some of the vertices from \(\cM\) back, in order to preserve also an approximate {\nice} {\cbtds}. We start by setting \(\cN = \emptyset\). 
Now for every set \(L\subseteq V(G)\setminus \cM \) of size at most \(t = 2^{\lceil \frac{1}{\delta}\rceil}\) we compute a Steiner tree \(T_L\) for the set of terminals \(L\) in \(G\). If \(T_L\) has at most \((1+\delta)k\) vertices, we add all vertices on \(T_L\) to \(\cN\). It follows from Lemma~\ref{lemma:nice_solution_after_marking_size} that \(|\cN|=  \bigoh(k^{(3d+6\eta)t + 1})\). Since \(t\) is a constant, it follows from Proposition~\ref{proposition:steiner-tree-terminals} that we can compute each of at most $\bigoh(k^{(3d+6\eta)t})$ Steiner trees in polynomial time. We now let $G' = G-(\cM\setminus \cN)$.

The following lemma will be useful to show that there is a small \nice\ \cbtds\ solution in \(G'\). Moreover, it will be also useful in our solution-lifting algorithm, where we need to first transform the solution to a \nice\ \cbtds.

\begin{lemma}\label{lemma:constructing_nice_connected_solution}
	Let \(Y\subseteq (V(G)\setminus (X\cup Z))\) and let \(S\) be a \cbtds\ for \(G-Y\) of size at most \(k\). 
	There is a polynomial-time algorithm that takes on the input \(G\), \(Y\), and \(S\) and outputs 
	a \nice\ \cbtds\ for \(G-Y\) of size at most \((1+\delta)|S|\).
\end{lemma}

\begin{proof}
	Let \(\cT=\{T_1, \ldots, T_r\}\) be the collection of distinct vertex sets such that for all \(i\in [r]\), \(T_i\subseteq (X\cup Z)\setminus H\), and \(S\) intersects \(\components{T_i}\) in more than \(\lambda\) vertices.
	If for all \(i\in [r]\) we have that \(T_i\subseteq S\), then \(S\) is \nice. Hence, it suffices to add all sets \(T_i\) to \(S\) and make the final set connected without adding to the solution a vertex that is not in a component of \(\components{T_i}\) for some \(T_i\in \cT\). To get a {\nice} {\cbtds} for \(G-Y\), let us start with \(S'=S\).
	For every \(i\in [r]\) we do the following. We first add \(T_i\) to \(S'\). Now, let \(C_i\) be a connected component of \(G-(X\cup Z\cup Y)\) in \(\components{T_i}\) such that \(C_i\cap S\) is not empty. Since \(S\) intersects \(\components{T_i}\) in more than \(\lambda\) vertices, such a connected component \(C_i\) exists. Let \(v\) be an arbitrary vertex in \(C_i\cap S\). Since \(C_i\) is a connected component of \(G-(X\cup Z\cup Y)\) it follows that \(\td(G[C_i])\le \eta\) and the diameter of \(G[C_i]\) is at most \(2^{\eta}\) by Proposition~\ref{lemma:bounded-treedepth-diameter}. Moreover, \(T_i\subseteq N(C_i)\) and for every vertex \(u\) in \(T_i\) there is an \(u\)-\(v\) path in \(G\) of length at most \(2^{\eta}\). We add to \(S'\) a shortest path from every vertex of \(T_i\setminus S\) to the vertex \(v\). It is straightforward to verify that \(S'\) is a \nice\ \cbtds for \(G-Y\). It only remains to show that \(|S'|\le (1+\delta)|S|\). For each \(i\in [r]\) we added to \(S\) at most \(|T_i\setminus S|\cdot 2^{\eta}\) vertices. Moreover, \(|T_i\setminus S|\le |T_i|\le d + 2\eta\). Therefore, in total we added at most \((d+2\eta)\cdot 2^{\eta}\cdot r\) vertices to \(S\). Now for \(i\neq j\) we have that \(\components{T_i}\) and \(\components{T_j}\) are pairwise disjoint. Since for each \(i\in [r]\), \(\components{T_i}\) contain at least \(\lambda+1\) vertices of \(S\), it follows that \(r\le \frac{|S|}{\lambda+1}\). Hence \(|S'|-|S| \le (d+2\eta)\cdot 2^{\eta} \cdot\frac{|S|}{\lambda+1}\le \delta |S|\).
\end{proof}

Using the above lemma, we now prove the following two lemmas that we will eventually use to prove our final theorem statement.
The proofs of the following two lemmas use the correctness of Lemma~\ref{lemma:constructing_nice_connected_solution}.

\begin{lemma}\label{lemma:nice_connected_solution_in_reduced}
	If $\OPT(G,k)\le k$, then 
	there exists a {\cbtds} for $G'$ of size at most \((1+\delta)^2\OPT(G,k)\).
\end{lemma}

\begin{proof}
	From Lemma~\ref{lemma:constructing_nice_connected_solution} applied on \(Y=\emptyset\) and an optimum solution \(S\), it follows that if \((G,k)\) is \Yes-instance, then there exists a \nice\ \cbtds\ \(S'\) for \(G\) of size at most \((1+\delta)\OPT(G,k)\). Now it is easy to see that \(S'\setminus \cM\) is a \nice\ \btds\ for \(G-\cM\) and by Lemma~\ref{lemma:nice_solution_after_marking} it is a \nice\ \btds\ for \(G\) as well. Now \(S'\) is a Steiner tree for \(S'\setminus \cM\) of size at most \((1+\delta)\OPT(G,k)\le (1+\delta)k\), hence the size of optimal Steiner tree for every subset \(L\subseteq S'\setminus \cM\) is also at most \((1+\delta)k\). Therefore, for every \(L\subseteq S'\setminus \cM\) of size at most \(t = 2^{\lceil \frac{1}{\delta}\rceil}\), \(G-(\cM\setminus\cN)\) contains an optimal Steiner tree \(T_L\) for the set of terminals \(L\) in \(G\). Hence, \(G-(\cM\setminus\cN)\) contains an optimal \(t\)-restricted Steiner tree \(\cT\) for \(S'\setminus \cM\). Clearly, the vertices of the \(t\)-restricted Steiner tree \(\cT\) for \(S'\setminus \cM\) induce a connected subgraph of \(G\) and contain all vertices in \(S'\setminus \cM\). Since \(S'\setminus \cM\) is a \nice\ \btds\ for \(G\), it follows that the vertices of \(\cT\) from a \cbtds. By Proposition~\ref{proposition:approx-k-steiner-tree}, the size of \(\cT\) is at most \((1+\frac{1}{\lfloor \log t \rfloor})|S'|\le (1+\delta)|S'|\le (1+\delta)^2\OPT(G,k)\).
\end{proof}

\begin{lemma}\label{lemma:solution_lifting}
	Given a {\cbtds} $S'$ of $(G', k')$ of size at most \(k\), we can in polynomial time compute a {\cbtds} \(S\) of $(G, k)$ such that $$\frac{|S|}{\OPT(G, k)} \le (1+\delta)^3 \frac{|S'|}{\OPT(G', k')}. $$
\end{lemma}

\begin{proof}
	By Lemma~\ref{lemma:constructing_nice_connected_solution}, we can in polynomial time construct a \nice\ \cbtds\ \(S\) for \(G'\) such that \(|S|\le (1+\delta)|S'|\). Since \(V(G)\setminus \cM \subseteq V(G')\), it follows from Lemma~\ref{lemma:nice_solution_after_marking} that \(S\) is also a \nice\ \cbtds\ for \(G\). From Lemma~\ref{lemma:nice_connected_solution_in_reduced}, it follows that 
	$$\frac{\OPT(G',k')}{\OPT(G,k)}\le (1+\delta)^2$$
	
	Combining the two inequalities, we get that $$\frac{|S|}{\OPT(G, k)} \le (1+\delta)^3 \frac{|S'|}{\OPT(G', k')}$$
\end{proof}

We are now ready to prove our main result.


\thmpsaksCBTDSmain*

\begin{proof}
{\color{black}
We choose $\delta, \lambda$ and $t$ as described earlier. 
}
The approximate kernelization algorithm has two parts, i.e. reduction algorithm and solution-lifting algorithm.
Let $(G, k)$ be an input instance of {\CBTDS}.

{\color{blue} \ding{226}} 
{\bf Reduction Algorithm:} If $G$ has two distinct components both having treedepth at least $\eta + 1$, we output $(K_{\eta + 2} \uplus K_{\eta + 2}, 1)$.
{\color{black} If $|V(G)| \leq 2^{3\eta^2 + d\eta}(\lambda + \eta + 1)^{\eta + 1}(1+\delta)k^{(3d + 6\eta)t + 1}$, then we output $(G, k)$.}
Otherwise, we do the following.
\begin{itemize}
	\item We first apply Reduction Rule~\ref{red-rule:redundant-component-removal} to remove all connected components of $G$ with treedepth at most $\eta$.
	\item Then, we invoke Lemma~\ref{lemma:decomposition-of-G} to compute a decomposition of $G$ such that $V(G) = X \cup Z \cup R$ and the conditions are satisfied.
	\item Then, we apply Reduction Rule~\ref{red-rule:forcing-neighborhood-of-C} exhaustively to construct $H$ and the output instance is $(G_1, k_1)$ with $k_1 = k$.
	\item We use Reduction Rules~\ref{rrule:components} and~\ref{rrule:single_component} on $(G_1, k_1)$ exhaustively to compute $\cM$. 
	
	\item Afterwards, we compute an optimal Steiner tree \(T_L\) for every subset \(L\) of \(V(G)\setminus \cM\) of size at most \(t= 2^{\ceil{\frac{1}{\delta}}}\) and if its size is at most \((1+\delta)k\), then we add \(T_L\) to the set $\cN$. 
	\item We delete $\cM \setminus \cN$ from $(G_1, k_1)$ to compute the instance $(G', k')$ with $k' = k_1$.
	\item Output $(G', k')$.
\end{itemize}

{\color{blue} \ding{226}}
{\bf Solution-lifting Algorithm:} Let $S'$ be a {\cbtds} of $(G', k')$.
Recall that $H \subseteq S'$.
If $|S'| > k'$, then we output the entire vertex set of a connected component of $G$ whose treedepth is larger than $\eta$.
If $|S'| \leq k'$, then $\CBT (G', k') = |S'|$.
We invoke Lemma~\ref{lemma:solution_lifting} to compute a nice {\cbtds} $S_1$ of the instance $(G_1, k_1)$ such that
 
\begin{equation}
\centering
\frac{|S_1|}{\OPT(G_1, k_1)} \leq (1+\delta)^3 \frac{\CBT (G', k', S')}{\OPT(G', k')}
\label{eq:final-proof-eq-1}	
\end{equation}

By construction, $H \subseteq S_1$.
If $|S_1| \leq k_1$, we output $S = S_1$ as a {\cbtds} of $(G, k)$.
If $|S_1| > k_1$, we output the entire vertex set of a connected component whose treedepth is larger than $\eta$.
We are giving the proof for $|S| \leq k$. The other case can be proved in a similar way.

As $k \geq |S| \geq \OPT(G, k)$, it follows from Lemma~\ref{lemma:alpha-safeness-red-rule-forcing-neighborhood-of-C} that $\OPT (G_1, k_1) \leq (1+\delta) \OPT(G, k)$.
Moreover, $\CBT(G, k, S) = |S|$.
Hence, using (\ref{eq:final-proof-eq-1}) we have the following.

	$$\frac{\CBT (G, k, S)}{\OPT (G, k)} \leq (1+\delta) \frac{\CBT (G_1, k_1, S_1)}{\OPT(G_1, k_1)}$$
	$$ = (1+\delta) \frac{|S_1|}{\OPT(G_1, k_1)} \leq (1+\delta)^4 \frac{\CBT (G', k', S')}{\OPT(G', k')}$$
	
As $(1+\delta)^4 \leq (1+\eps)$, a $c$-approximate {\cbtds} of $(G', k')$ can be lifted to a $c(1+\eps)$-approximate {\cbtds} of $(G, k)$.

On the other hand, if $|S'| > k$, then we output the entire connected component $C$ of $V(G)$ with treedepth at least $\eta + 1$.
Observe that $\CBT (G, k, C) = k + 1$ and $k = k_1 = k'$.
Moreover, $\CBT (G', k', S') = k' + 1$.
Here are the following cases.
\begin{description}
	\item[Case 1:] Let $\OPT (G, k) = k+1$. Then,
	
	$$\frac{\CBT (G, k, C)}{\OPT(G, k)} = \frac{k+1}{\OPT(G, k)} = 1 \leq (1+\eps) \frac{\OPT(G', k', S')}{\OPT(G',k')}$$
	
	\item[Case 2:]
Let $\OPT(G, k) \leq k$. Then $\OPT(G_1, k_1) \leq (1+\delta)\OPT(G, k)$.
Furthermore, if $\OPT(G_1, k_1) \leq k_1$, then $\OPT(G', k') \leq (1+\delta)^3 \OPT(G_1, k_1)$.
Therefore, $\OPT(G', k') \leq (1+\delta)^4 \OPT(G, k) \leq (1+\eps)\OPT(G, k)$ and we have the following.

$$\frac{\CBT (G, k, C)}{\OPT(G, k)} = \frac{k+1}{\OPT(G, k)} \leq (1+\eps) \frac{\CBT (G', k', S')}{\OPT(G', k')}$$

Otherwise, let $\OPT(G_1, k_1) = k+1$.
Then, for any (nice) {\cbtds} $S_1$ of $(G_1, k_1)$, it holds that $$\frac{\CBT (G, k, C)}{\OPT(G, k)} \leq (1+\delta)\frac{\CBT (G_1, k_1, S_1)}{\OPT(G_1, k_1)}$$
$$ = (1+\delta) \leq (1+\eps)\frac{\CBT (G', k', S')}{\OPT(G', k')}$$
\end{description}

By construction, Lemma~\ref{lemma:nice_solution_after_marking_size}, and the size bound of $|\cN|$, we have that $|V(G')| = \cO(k^{(3d + 6\eta) t + 1})$.
Recall that $d = \ceil{2^{\eta + 3}\eta/\delta}$, $t = 2^{\ceil{1/\delta}}$, and \(\delta = \frac{\eps}{10}\).
{\color{black} Observe that all the reduction rules can be performed in $k^{\cO(3d + 6\eta)t} n^{\cO(1)}$-time. As these reduction rules are executed only when $|V(G)| = n > 2^{3\eta^2 + d\eta}(\lambda + \eta + 1)^{\eta + 1}(1+\delta)k^{(3d + 6\eta)t + 1}$, this PSAKS is time efficient}.
The reduction algorithm and the solution-lifting algorithm run in polynomial time.
Therefore, we have a time-efficient $(1+\eps)$-approximate kernel with the claimed bound.
\end{proof}

\section{Conclusions}
We have obtained a polynomial-size approximate kernelization scheme (PSAKS) for {\CBTDS}, improving upon existing bounds and advancing the line of work on approximate kernels for vertex deletion problems with connectivity constraints.  Towards our result, we combined known decomposition techniques with new preprocessing steps that exploit structure present in bounded treedepth graphs. 
%
Our work points to a few interesting questions for follow up research:

\begin{enumerate}
\item Is there a PSAKS for {\sc $\eta$-Treedepth Deletion} with stronger connectivity constraints, e.g., when the solution is required to induce a biconnected graph or a $2$-edge-connected graph. Recently, Einarson et al.~\cite{EGJMW20arxiv} initiated this line of research in the context of studying approximate kernels for {\sc Vertex Cover} with stronger connectivity constraints. Could a similar result be obtained for {\sc $\eta$-Treedepth Deletion} with stronger connectivity requirements?
\item Would it be possible to improve the size of our PSAKS to $f(\eta)\cdot k^{\bigoh(1/\eps)}$, i.e. the exponent of $k$ becomes independent of $\eta$? It is known that {\sc $\eta$-Treedepth Deletion} problem (without connectivity constraints) admits a kernel with $\cO(2^{\cO(\eta^2)}k^6)$ vertices~\cite{GiannopoulouJLS17}. In fact, parts of our algorithm are based on this work of Giannopoulou et al. However,  we incur the $k^{\bigoh(2^{\bigoh(\eta)}\cdot 1/\eps)}$ cost in the size of our kernel in several places (e.g. Reduction Rules~\ref{rrule:components} and \ref{rrule:single_component}).
We believe that one would need to formulate a significantly distinct approach in order to attain a bound of $f(\eta)\cdot k^{\bigoh(1/\eps)}$.
Obtaining such a uniform PSAKS for {\CBTDS} is an interesting open problem.
\item Could one get a PSAKS for {\sc Connected $\eta$-Treewidth Deletion}? The current best approximate kernelization result for this problem is the $(2+\eps)$-approximate polynomial compression from \cite{Ramanujan21}.
We believe that several parts of our algorithm can be adapted to work for {\sc $\eta$-Treewidth Deletion}. However, we have crucially used the fact that a connected bounded treedepth graph has bounded diameter, which is a property one cannot assume for bounded treewidth graphs. 

\end{enumerate}

\mypara{Acknowledgement:} Research of M. S. Ramanujan has been supported by Engineering and Physical Sciences Research Council (EPSRC) grants EP/V007793/1 and EP/V044621/1.



\begin{thebibliography}{10}

\bibitem{BodlaenderJK14}
Hans~L. Bodlaender, Bart M.~P. Jansen, and Stefan Kratsch.
\newblock Kernelization lower bounds by cross-composition.
\newblock {\em {SIAM} J. Discrete Math.}, 28(1):277--305, 2014.

\bibitem{BorchersD97}
Al~Borchers and Ding{-}Zhu Du.
\newblock The k-steiner ratio in graphs.
\newblock {\em {SIAM} J. Comput.}, 26(3):857--869, 1997.

\bibitem{ByrkaGRS13}
Jaroslaw Byrka, Fabrizio Grandoni, Thomas Rothvo{\ss}, and Laura Sanit{\`{a}}.
\newblock Steiner tree approximation via iterative randomized rounding.
\newblock {\em J. {ACM}}, 60(1):6:1--6:33, 2013.

\bibitem{Cygan12}
Marek Cygan.
\newblock Deterministic parameterized connected vertex cover.
\newblock In {\em Algorithm Theory - {SWAT} 2012 - 13th Scandinavian Symposium
  and Workshops, Helsinki, Finland, July 4-6, 2012. Proceedings}, pages
  95--106, 2012.

\bibitem{CyganFKLMPPS15}
Marek Cygan, Fedor~V. Fomin, Lukasz Kowalik, Daniel Lokshtanov, D{\'{a}}niel
  Marx, Marcin Pilipczuk, Michal Pilipczuk, and Saket Saurabh.
\newblock {\em Parameterized Algorithms}.
\newblock Springer, 2015.

\bibitem{CyganPPW13}
Marek Cygan, Marcin Pilipczuk, Michal Pilipczuk, and Jakub~Onufry Wojtaszczyk.
\newblock Subset feedback vertex set is fixed-parameter tractable.
\newblock {\em {SIAM} J. Discret. Math.}, 27(1):290--309, 2013.

\bibitem{DiestelBook}
Reinhard Diestel.
\newblock {\em Graph Theory, 4th Edition}, volume 173 of {\em Graduate texts in
  mathematics}.
\newblock Springer, 2012.

\bibitem{DomLS14}
Michael Dom, Daniel Lokshtanov, and Saket Saurabh.
\newblock Kernelization lower bounds through colors and {ID}s.
\newblock {\em {ACM} Trans. Algorithms}, 11(2):13:1--13:20, 2014.

\bibitem{DreyfusW71}
S.~E. Dreyfus and R.~A. Wagner.
\newblock The steiner problem in graphs.
\newblock {\em Networks}, 1(3):195--207, 1971.

\bibitem{DvorakGT12}
Zdenek Dvor{\'{a}}k, Archontia~C. Giannopoulou, and Dimitrios~M. Thilikos.
\newblock Forbidden graphs for tree-depth.
\newblock {\em Eur. J. Comb.}, 33(5):969--979, 2012.

\bibitem{EibenHR19}
Eduard Eiben, Danny Hermelin, and M.~S. Ramanujan.
\newblock On approximate preprocessing for domination and hitting subgraphs
  with connected deletion sets.
\newblock {\em J. Comput. Syst. Sci.}, 105:158--170, 2019.

\bibitem{EibenMR22}
Eduard Eiben, Diptapriyo Majumdar, and M.~S. Ramanujan.
\newblock On the lossy kernelization for connected treedepth deletion set.
\newblock In Michael~A. Bekos and Michael Kaufmann, editors, {\em
  Graph-Theoretic Concepts in Computer Science - 48th International Workshop,
  {WG} 2022, T{\"{u}}bingen, Germany, June 22-24, 2022, Revised Selected
  Papers}, volume 13453 of {\em Lecture Notes in Computer Science}, pages
  201--214. Springer, 2022.

\bibitem{EGJMW20arxiv}
Carl Einarson, Gregory~Z. Gutin, Bart M.~P. Jansen, Diptapriyo Majumdar, and
  Magnus Wahlstr{\"{o}}m.
\newblock p-edge/vertex-connected vertex cover: Parameterized and approximation
  algorithms.
\newblock {\em CoRR}, abs/2009.08158, 2020.

\bibitem{FominLMS12}
Fedor~V. Fomin, Daniel Lokshtanov, Neeldhara Misra, and Saket Saurabh.
\newblock {Planar F-Deletion: Approximation, Kernelization and Optimal {FPT}
  Algorithms}.
\newblock In {\em 53rd Annual {IEEE} Symposium on Foundations of Computer
  Science, {FOCS} 2012, New Brunswick, NJ, USA, October 20-23, 2012}, pages
  470--479. {IEEE} Computer Society, 2012.

\bibitem{ford1962flows}
Lester~R. Ford and Delbert~R. Fulkerson.
\newblock {\em Flows in Networks}.
\newblock Princeton University Press, 1962.

\bibitem{GajarskyHOORRVS17}
Jakub Gajarsk{\'{y}}, Petr Hlinen{\'{y}}, Jan Obdrz{\'{a}}lek, Sebastian
  Ordyniak, Felix Reidl, Peter Rossmanith, Fernando~S{\'{a}}nchez Villaamil,
  and Somnath Sikdar.
\newblock Kernelization using structural parameters on sparse graph classes.
\newblock {\em J. Comput. Syst. Sci.}, 84:219--242, 2017.

\bibitem{GiannopoulouJLS17}
Archontia~C. Giannopoulou, Bart M.~P. Jansen, Daniel Lokshtanov, and Saket
  Saurabh.
\newblock Uniform kernelization complexity of hitting forbidden minors.
\newblock {\em {ACM} Trans. Algorithms}, 13(3):35:1--35:35, 2017.

\bibitem{HegerfeldK20}
Falko Hegerfeld and Stefan Kratsch.
\newblock Solving connectivity problems parameterized by treedepth in
  single-exponential time and polynomial space.
\newblock In Christophe Paul and Markus Bl{\"{a}}ser, editors, {\em 37th
  International Symposium on Theoretical Aspects of Computer Science, {STACS}
  2020, March 10-13, 2020, Montpellier, France}, volume 154 of {\em LIPIcs},
  pages 29:1--29:16. Schloss Dagstuhl - Leibniz-Zentrum f{\"{u}}r Informatik,
  2020.

\bibitem{HermelinKSWW15}
Danny Hermelin, Stefan Kratsch, Karolina Soltys, Magnus Wahlstr{\"{o}}m, and
  Xi~Wu.
\newblock A completeness theory for polynomial (turing) kernelization.
\newblock {\em Algorithmica}, 71(3):702--730, 2015.

\bibitem{HermelinW12}
Danny Hermelin and Xi~Wu.
\newblock Weak compositions and their applications to polynomial lower bounds
  for kernelization.
\newblock In {\em Proceedings of the Twenty-Third Annual {ACM-SIAM} Symposium
  on Discrete Algorithms, {SODA} 2012, Kyoto, Japan, January 17-19, 2012},
  pages 104--113, 2012.

\bibitem{JansenP20}
Bart M.~P. Jansen and Astrid Pieterse.
\newblock Polynomial kernels for hitting forbidden minors under structural
  parameterizations.
\newblock {\em Theor. Comput. Sci.}, 841:124--166, 2020.

\bibitem{LokshtanovPRS17}
Daniel Lokshtanov, Fahad Panolan, M.~S. Ramanujan, and Saket Saurabh.
\newblock Lossy kernelization.
\newblock In {\em Proceedings of the 49th Annual {ACM} {SIGACT} Symposium on
  Theory of Computing, {STOC} 2017, Montreal, QC, Canada, June 19-23, 2017},
  pages 224--237, 2017.

\bibitem{MisraPRS14}
Neeldhara Misra, Geevarghese Philip, Venkatesh Raman, and Saket Saurabh.
\newblock The kernelization complexity of connected domination in graphs with
  (no) small cycles.
\newblock {\em Algorithmica}, 68(2):504--530, 2014.

\bibitem{NesetrilM06}
Jaroslav Nesetril and Patrice~Ossona de~Mendez.
\newblock Tree-depth, subgraph coloring and homomorphism bounds.
\newblock {\em Eur. J. Comb.}, 27(6):1022--1041, 2006.

\bibitem{Ramanujan19}
M.~S. Ramanujan.
\newblock {An Approximate Kernel for Connected Feedback Vertex Set}.
\newblock In {\em 27th Annual European Symposium on Algorithms, {ESA} 2019,
  September 9-11, 2019, Munich/Garching, Germany}, pages 77:1--77:14, 2019.

\bibitem{Ramanujan21}
M.~S. Ramanujan.
\newblock {On Approximate Compressions for Connected Minor-hitting Sets}.
\newblock In {\em 29th Annual European Symposium on Algorithms, {ESA} 2021},
  2021.

\bibitem{ReidlRVS14}
Felix Reidl, Peter Rossmanith, Fernando~S{\'{a}}nchez Villaamil, and Somnath
  Sikdar.
\newblock {A Faster Parameterized Algorithm for Treedepth}.
\newblock In Javier Esparza, Pierre Fraigniaud, Thore Husfeldt, and Elias
  Koutsoupias, editors, {\em Automata, Languages, and Programming - 41st
  International Colloquium, {ICALP} 2014, Copenhagen, Denmark, July 8-11, 2014,
  Proceedings, Part {I}}, volume 8572 of {\em Lecture Notes in Computer
  Science}, pages 931--942. Springer, 2014.

\end{thebibliography}

\end{document}